\documentclass[aps,showpacs]{revtex4}

\usepackage{amssymb}
\usepackage{amsmath}
\usepackage{amsthm}

\newtheorem{thm}{Theorem}
\newtheorem{Def}{Definition}
\newtheorem*{thmmj}{Theorem (Morrow-Jones and Witt \cite{MorrowJones:1993zu})}

\begin{document}

\title{Designer de Sitter Spacetimes}

\author{K.\ Schleich}
\affiliation{Department of Physics and Astronomy, University of British Columbia,
Vancouver, British Columbia \ V6T 1Z1}
\author{D.M.\ Witt}
\affiliation{Department of Physics and Astronomy, University of British Columbia,
Vancouver, British Columbia \ V6T 1Z1}

\begin{abstract}
Recent observations in cosmology indicate an accelerating expansion of the universe postulated to arise from some form of dark energy, the paradigm being positive cosmological constant.
De Sitter spacetime is the well-known isotropic solution to the Einstein equations
with cosmological constant. However, as discussed here, it is not the most general, {\it locally} isotropic solution. One can construct an infinite family of such solutions, designer de Sitter spacetimes, which are everywhere locally isometric to a region of
de Sitter spacetime. However, the global dynamics of these designer cosmologies is very different
than that of de Sitter spacetime itself. The construction and dynamics of these designer de Sitter spacetimes is detailed along with some comments about their implications for the structure of our universe. 

\end{abstract}
\maketitle
\section{Introduction}
One goal of cosmology is to determine as fully as possible the physics of our universe. The standard approach toward this goal is to describe our universe in terms of a cosmological model dependent on a few simple parameters.  These parameters are then constrained by local observations, more precisely, observations of physical observables such as type I supernova \cite{Perlmutter:1998np}, WMAP data \cite{Spergel:2003cb} and the 3-dimensional galaxy power spectrum from SDSS \cite{Tegmark:2003ud}.
These parameters are inferred to have the same values everywhere through invocation of some form of the Cosmological Principle; namely that all points of the universe are in some sense equivalent on sufficiently large scales. Indeed, this assumption is built into the standard cosmological model through its use of the Robertson-Walker metric, a homogeneous, isotropic spacetime.  Global isotropy dictates that the metric has the form $ds^2 = - dt^2 + a^2(t)d\sigma_3^2$  where $d\sigma_3^2$ is one of three spatial metrics with constant sectional curvature; the round metric, the flat metric and the hyperbolic metric.  The three spatial geometries determine the spatial topology to be either $S^3$ or ${\mathbb R}^3$. The
evolution of the spatial geometry is just given by the scale factor $a(t)$. Hence, in the context of the standard cosmological model, the global dynamics of the universe is determined from the local measurements of a few parameters, in particular $\Omega$, the ratio of
the current average energy density to the critical energy density.

The application of this approach has been one of the great successes of the last decade; in the context of this cosmological model, data now constrains the parameters characterizing our universe to high precision.  In fact, it points to an accelerating expansion of the universe due to a large dark energy contribution to $\Omega$ of nearly three quarters of the critical density  with positive cosmological constant  $\Lambda$ the most likely source \cite{Spergel:2006hy}. Our universe is now described as a $\Lambda$ CDM model with a strong likelihood of an inflationary period. Thus the dynamics of spacetimes with cosmological constant are now clearly of key physical significance.

The paradigm of spacetimes with positive cosmological constant is de Sitter spacetime,  the well-known solution of $R_{ab} =\Lambda g_{ab}$. Its  metric can be written in each of the three Robertson-Walker forms:
\begin{equation}\label{sphere} ds^2 =  -dt^2 + \alpha^2 \cosh (t/\alpha )
	\left( d\psi^2 + \sin^2 \psi
	d\Omega _2 ^2  \right)\end{equation}
	with Cauchy slice $ S^3$,
\begin{equation}\label{flat} ds^2 = -dt^2 + \alpha ^2 e^{ 2 t/\alpha }
	\left( d\psi ^2 + \psi ^2
	 d\Omega _2 ^2  \right)\end{equation}
	  with partial Cauchy slice $ {\mathbb R^3}$ and
\begin{equation} \label{hyp} ds^2 = -dt^2 + \alpha ^2 \sinh ^2 \left( t/\alpha \right)
	\left( d\psi ^2 + \sinh ^2 \psi 
	 d\Omega _2 ^2 \right)\end{equation}
	 with partial Cauchy slice ${\mathbb R^3}$ where $\alpha = \sqrt{3/\Lambda}$ and  $d\Omega _2 ^2$ is the round metric on the 2-sphere. This globally isotropic spacetime, the basis for inflationary models and used in the calculation of the power spectrum from the inflationary period, hence satisfies forms of the Cosmological Principle that assume global isotropy. 
	 
	Given the importance of de Sitter spacetime to the standard cosmological model, it is natural to ask if other solutions to $R_{ab} =\Lambda g_{ab}$ exist. In particular, is global isotropy required by the dynamics in spacetimes with cosmological constant? Although Robertson-Walker cosmologies are always taken to be globally isotropic, it is interesting to note that their original motivation was only local isotropy, not global isotropy. In one of the original papers on the
Robertson-Walker cosmologies \cite{Walker1937}, Walker assumes that the every point has a local set of
coordinates such that the spacetime metric about that point can be written as an
isotropic metric. This distinction is generally neglected: the folklore in general relativity is that locally isotropic spacetimes are simply identifications on the globally isotropic ones. Thus, according to folklore,
 the global dynamics of locally isotropic cosmologies is still determined in terms of a single scale factor. 
 
 However, this folklore is false for de Sitter spacetimes;  one can construct initial data sets for locally isotropic de Sitter spacetimes whose covering space is not one of the three globally isotropic de Sitter metrics. We use the results of Morrow-Jones and Witt \cite{MorrowJones:1993zu} to demonstrate that the global dynamics of certain locally isotropic de Sitter spacetimes, designer de Sitter spacetimes, are in fact not determined by a single scale factor.  These designer de Sitter spacetimes have generic topology and consequently are abundant in the set of all locally isotropic solutions of the Einstein equations with positive cosmological constant.  We conclude by qualitatively discussing the potential signatures of these spacetimes on the angular power map spectrum of the WMAP data.
 
 \section{Designer de Sitter Spacetimes}

Historically, the existence of locally isotropic de Sitter spacetimes of generic topology was first proved by  Morrow-Jones and Witt, appearing in both the Ph.D. thesis of Morrow-Jones \cite{MorrowJones:1988yw} and in journal form in \cite{MorrowJones:1993zu}. Later, mathematicians became 
interested in the same type of spacetimes in 2+1 dimensions \cite{mess} because of connections 
between the spacetime structures and conformal structures and the work of Thurston on 3-manifolds.  
A spacetime approach to the Morrow-Jones and Witt solutions was presented in 
\cite{Bengtsson:1999ia}.
The results of  \cite{MorrowJones:1993zu} in 3+1 dimensions were then later reproduced without reference in \cite{student}.
  
The Einstein equations for a globally hyperbolic spacetime  $\Sigma \times {\mathbb R}$ with
4-metric $g_{ab}$ can be written in initial value form on the Cauchy slice $\Sigma $ by introducing
coordinates such that $\Sigma$ is a constant time slice in the globally hyperbolic spacetime.
The normal vector to $\Sigma$ is  
\begin{equation*}n^a = \frac {1}{N}(t^a - N^a)\end{equation*}
 where $t^a$ is the
tangent vector to the time function. $N$ is usually termed the lapse and $N^a$ the shift. Then
the metric 
\begin{equation*} g_{ab} = -n_a n_b + h_{ab}\end{equation*}
where $h_{ab}$ is a geodesically complete
Riemannian metric on $\Sigma$. The extrinsic curvature of $\Sigma$  is given by
\begin{equation*}K_{ab}=\frac 12 {\mathcal L}_n{h_{ab}}\ .\end{equation*}
The 
3-metric is the restriction of the spacetime metric to $\Sigma$ and extrinsic curvature is the 
rate of change of the 3-metric in the spacetime. 
The 3-metric and extrinsic 
curvature satisfy the constraints:
\begin{align}&R-K_{ab}K^{ab}+K^2=16\pi \rho\cr
&D_b(K^{ab}-Kh^{ab})=8\pi J^a\label{constraints}\end{align}
where $R$ is the scalar curvature of $h$, and $\rho$ and $J^a$
are the energy and momentum densities of the matter sources. The remaining 
 Einstein equations give the time evolution of $h$ and $K$. 

Conversely, given any 3-manifold $\Sigma$ with geodesically complete Riemannian metric $h_{ab}$ and symmetric tensor 
$K_{ab}$ satisfying the constraints (\ref{constraints}) with physically reasonable matter sources (precisely, matter sources that satisfy the dominant energy condition $\rho\geq (J^aJ_a)^{\frac 12}$), one can the evolve the initial data $h_{ab}$ and $K_{ab}$ into a globally
hyperbolic spacetime  $\Sigma \times {\mathbb R}$. 

Next,  the notions of locally symmetric tensor fields
and spacetimes with local symmetries are made precise:

\begin{Def}\label{lss1} A tensor field, $T_{ab\cdots }{}^{cd\cdots }$,
on a manifold, $M$, is {\it locally symmetric with respect to} $G$ iff 
every point in $M$ has an open neighborhood $U$ such that the following 
conditions are satisfied:

(i) There is a finite set of vector fields $\{\xi _i{}^a\}$ on $U$ which 
generate a faithful representation of the Lie algebra of $G$. 

(ii) ${\mathcal L} _{\xi _i}T_{ab\cdots }{}^{cd\cdots } \bigg| _U = 0$ for these 
vectors. 
\end{Def}

If the tensor with the local symmetry is the metric on the manifold, then the 
vectors are local Killing vectors. The particular symmetry of interest in this paper
is local spherical symmetry. 
\begin{Def} \label{lss3} A tensor field, $T_{ab\cdots }{}^{cd\cdots }$,
on a manifold, $M$, is  locally spherically symmetric iff it is locally 
symmetric with respect to $SO(3)$ and the orbits of the vector fields are 
two dimensional. 
A manifold, $M$, with metric, $g_{ab}$,
is locally spherically symmetric if the metric is locally
spherically symmetric.\end{Def}
\begin{Def} \label{lss4} A manifold, $M$, with metric, $g_{ab}$,
is locally isotropic if it is locally spherically symmetric about every point. \end{Def}

A natural extension of these definitions is to initial data. Initial data is locally
spherically symmetric if both $h_{ab}$ and $K_{ab}$ are locally
spherically symmetric tensors on the Cauchy slice $\Sigma $.
 It follows from definitions \ref{lss1}-\ref{lss3} that any globally spherically symmetric
space must also be locally spherically symmetric. However, 
converse is not true in general.

For example, the Robertson-Walker cosmologies are globally isotropic only if no global
identifications are allowed. If global identifications are allowed, the number of
global symmetries is reduced and the spacetime fails to be globally isotropic. A
simple example of this is periodic identifications on the topology ${\mathbb R}^3$ with flat metric;
the resulting space is a 3-torus, $T^3$. The periodic translations leave
the flat metric on ${\mathbb R}^3$ invariant. Therefore, $T^3$ is locally isotropic. 
However, it is not
globally isotropic. The reason is that the spatially flat metric on $T^3$ has lost
its global rotational symmetries. In particular, although there is a neighborhood about every point  in $T^3$ that is explicitly rotationally invariant, that neighborhood cannot be extended across the boundaries of the fundamental domain of the identifications as the rotational killing vector becomes multiply defined under the identifications.

One can take any two 3-manifolds $\Sigma_1$ and $\Sigma_2$ and topologically join them by removing a 3-ball from each, then identifying the resulting $S^2$ boundaries. The resulting
space $\Sigma_1 \# \Sigma_2$ is called 
the {\it connected sum}. Taking the connected sum of two manifolds is a
topological operation and does not necessarily preserve any geometric properties. However, the theorem below shows that the connected sum can indeed preserve local isotropy.
 The solutions of these initial data sets are locally de Sitter solutions; every observer can find a local coordinate system in which the spacetime is manifestly a de Sitter cosmology. 

\begin{thmmj} Let $(\Sigma _i, {h_{ab}} _i, {K_{ab}} _i)$  be any countable collection of  locally isotropic initial data 
sets for globally hyperbolic spacetimes of the form $M=\Sigma _i\times {\mathbb R}$
that satisfy $R_{ab}=\Lambda g_{ab}$ with $\Lambda >0$.  Then one can
construct a locally isotropic solution to these equations with Cauchy slice 
$\Sigma = \Sigma _1\# \Sigma _2\# \Sigma _3\#   \dots   \#\Sigma _k\# \dots \ \  .$
Furthermore, restricted to each factor $\Sigma_i$, the initial data for this solution is just $({h_{ab}} _i, {K_{ab}} _i)$.\end{thmmj} 

\begin{proof} The proof, given fully in \cite{MorrowJones:1993zu} is outlined here.
Let $(\Sigma _1, {h_{ab}} _1, {K_{ab}} _1)$ and $(\Sigma _2, {h_{ab}} _2, {K_{ab}} _2)$ 
be any two 3-manifolds with locally isotropic initial data satisfying the Einstein equations 
$R_{ab}=\Lambda g_{ab}$ and $\Lambda >0$. Note that 
local isotropy implies that  one can construct local spherically symmetric
coordinates.
Next, define a metric on 
$\Sigma _1\# \Sigma _2$ by using local spherically symmetric coordinates in each
of the two manifolds in a neighborhood of the excised 3-balls, then smoothly deforming between the two metrics using
smooth, spherically symmetric bump functions. The resulting metric is manifestly locally spherically symmetric. Next, the extrinsic curvature ${K_{ab}}$ in this region can be constructed by integrating the constraints (\ref{constraints}). This construction also results in a locally spherically symmetric extrinsic curvature. Thus there exists locally spherically symmetric initial data on $\Sigma _1\# \Sigma _2$.  Next, the evolution of this initial data is shown to exist; furthermore, the resulting spacetime is proven to also be locally spherically symmetric. Next, using the cosmic Birkhoff theorem, it is shown that the spacetime on $\Sigma _1\# \Sigma _2$ is locally isotropic, that is it is locally de Sitter. Finally, iteration of the construction results in a globally hyperbolic spacetime with Cauchy slice of the form $\Sigma = \Sigma _1\# \Sigma _2\# \Sigma _3\#   \dots   \#\Sigma _k\# \dots \ \  $ that is locally de Sitter.
\end{proof}

Many different 3-manifolds are known to admit initial data sets for de Sitter spacetimes. For example, the well-known de Sitter solutions (\ref{sphere}-\ref{hyp})  have initial data sets on $\Sigma = {\mathbb R}^3$ (with both flat and hyperbolic spatial geometry) and $\Sigma= S^3$. Furthermore, there is also a Kasner slicing of de Sitter spacetime (see  \cite{MorrowJones:1993zu}) with topology $\Sigma = S^2 \times {\mathbb R}$ and metric
\begin{equation}\label{kasnerslice} ds^2 = -dt^2 + \alpha ^2 \left( \sinh ^2 \left( t/\alpha \right) d\psi ^2
	+ \cosh ^2 \left( t/\alpha \right) d\Omega _2 ^2
			\right)\ .\end{equation}	
	 In addition, identifications by a finite group $\Gamma$, for example $S^3/\Gamma$ or $\Sigma = {\mathbb R}^3/\Gamma $ for the case of hyperbolic and flat geometries, also admit locally isotropic initial data sets. Furthermore, arbitrary connected sums of these manifolds are generic 3-manifolds. Hence locally isotropic de Sitter initial data sets are generic in the set of all 3-manifolds.
	  
Given that generic 3-manifolds are Cauchy slices for locally de Sitter spacetimes,   it is natural to consider the physical
consequences of these solutions. 
Generically, these solutions are not the standard de Sitter
cosmologies or identifications thereof and have very different global dynamics. Specifically, there exists a class of locally de Sitter spacetimes, henceforth termed designer de Sitter spacetimes;

\begin{Def} \label{lss5} A designer de Sitter spacetime is any locally de Sitter spacetime with Cauchy slice of
the form $\Sigma = \Sigma _1\# \Sigma _2\# \Sigma _3\#   \dots   \#\Sigma _k\# \dots \ \  $
where  at least two $\Sigma _i$'s  are not $S^3$'s or at least one $\Sigma_i$ is $S^2\times {\mathbb R}$ or an identification on $S^2\times {\mathbb R}$.  \end{Def}
Observe that $S^3$ is the identity under the connected sum: $\Sigma\# S^3= \Sigma$ for any $\Sigma$. Therefore definition \ref{lss5} requires two of the factors to be other than an $S^3$ to ensure $\Sigma$ not  have a simple topology. The alternate requirement of at least one $S^2\times {\mathbb R}$ or an identification thereof will also suffice as seen below.

We begin by reviewing the construction of the {\it universal covering manifold} (also termed the {\it universal covering space}) ${\cal M}$ of a smooth
path-connected manifold $M$.  Pick a point $x_0\in
M$ and consider the set of smooth paths $P=\{c:[0,1]\rightarrow M|c(0)=x_0\}$.  A
projection map $\pi:P\rightarrow M$ is defined by  $\pi(c(t))=c(1)$.  Let ${\cal
M}$ be $P$ modulo the equivalence relation, $c_1\sim c_2$ if and only if
$c_1(1)=c_2(1)$ and $c_1$ is homotopic to $c_2$ with endpoints fixed. The
projection map $\pi$ is then well defined and smooth as a map $\pi:  {\cal
M}\rightarrow M$. By construction, the universal covering manifold ${\cal M}$ is simply connected.

We now prove that the global dynamics of a designer de Sitter spacetime cannot be described in terms of a single scale factor.
We begin by answering the question, what type of global topological structure can a
locally isotropic spacetime possess if one assumes that its global dynamics are obtained from a single scale factor $a(t)$?
\begin{thm} \label{firstthm} Let  $\Sigma \times {\mathbb R}$ be any globally hyperbolic spacetime
that satisfies $R_{ab}=\Lambda g_{ab}$ with $\Lambda >0$ whose global dynamics is given by
a single time varying scale factor $a(t)$. Then the $\Sigma $
has  either $S^3$ or ${\mathbb R}^3$ as its universal covering manifold. Furthermore, the spatial geometry 
is either spherical ($S^3$) or  flat or hyperbolic (${\mathbb R}^3$).
\end{thm} 
\begin{proof} If the dynamics of the spacetime is given 
by an overall scale factor, then in a neighborhood of the Cauchy slice $\Sigma$ the spacetime metric, written in synchronous gauge, yields  initial data with spatial metric of form
$h_{ab}(t)=a(t)h_{ab}$. Additionally, as $K_{ab}(t)=\frac 12 {\dot h}_{ab}(t)$ in this gauge, the extrinsic curvature also takes the form $K_{ab}(t)=b(t)h_{ab}$. 
Now, the 3-manifold $\Sigma $ with metric $h_{ab}(t)$ is complete. In addition, as there
is no Weyl curvature in 3 dimensions, the Riemann curvature is determined by the 
Ricci curvature $R_{ab}(t)$. The Einstein equations yield the 
following expression for the Ricci curvature of $h_{ab}(t)$ \cite{Wald};
\begin{equation}R^a_b - \frac {1}{2}R h^a_b + {\mathcal L}_{n} (K^a_b-h^a_bK) -KK^a_b+ 
 \frac {1}{2} h^a_bK^2-\frac {1}{2}h^a_bK^{cd}K_{cd}= -\Lambda \, h^a_b \  .\end{equation}
Substitution of the forms for $h_{ab}(t)$ and $K_{ab}(t)$ into the above expression yields
\begin{equation*}R_{ab}(h)=C(t)h_{ab}\ ,\end{equation*}
where $C(t)$ is only a function of $t$. Hence the Ricci scalar determines the Ricci curvature. It follows that the Riemann curvature is of the form 
\begin{equation*}R_{abcd}= D(t)(g_{ab}g_{cd}-g_{ac}g_{bd})\end{equation*}
where $D(t)$ is only a function of $t$. Thus, $\Sigma $ with family of metrics is $h_{ab}(t)$ is
a space form for all time. Finally, note that classic results from differential geometry \cite{wolf} imply
that if $\Sigma $ is a space form, it  has a universal covering space $S^3$ or ${\mathbb R}^3$. Moreover, the geometry of
the covering space is  either spherical, flat or hyperbolic.
\end{proof}
Theorem \ref{firstthm} implies that the spatial topology of any locally de Sitter spacetime with evolution described by a single scale factor is constructed by identifications on either the $S^3$ Cauchy slice or one of the two ${\mathbb R}^3$ partial Cauchy slices of the global de Sitter spacetime solution itself. Note in particular, the de Sitter spacetime  with Cauchy slice $S^2 \times {\mathbb R}$ (\ref{kasnerslice}) has an evolution described by two scale factors. Hence this topology is distinct.

We now use this result to prove the first property of designer de Sitter spacetimes:

\begin{thm} The global dynamics of designer de Sitter spacetimes is not
determined by a single scale factor. 
\end{thm} 

\begin{proof} Pick a designer de Sitter spacetime. 
If $\Sigma $ has its dynamics determined by a single scale factor, then according to theorem \ref{firstthm},
the universal covering manifold of $\Sigma $ is either $S^3$ or ${\mathbb R}^3$. Both $S^3$ and
${\mathbb R}^3$ are irreducible manifolds, namely every 2-sphere contained in them
is the boundary of a 3-ball \cite{Hemple}. In addition,  any space constructed by 
identifying points via a group action on either $S^3$ or ${\mathbb R}^3$ is also irreducible.
Hence, $\Sigma $ is irreducible if its dynamics is determined by a single scale factor. Next, if there are at least two factors which are not 3-spheres, there is a 
2-sphere in the manifold $\Sigma $ which is not the boundary of any 3-ball: the boundary $S^2$ identified in the connected sum. 
If there is at least one factor  of $S^2 \times {\mathbb R}$, the 2-sphere is obvious.
Having a 2-sphere that is not the boundary of any 3-ball  implies
 $\Sigma $ for designer de Sitter spacetimes is not irreducible. Hence, $\Sigma $ can not have universal covering space
$S^3$ or ${\mathbb R}^3$.  Therefore, the global dynamics is not given by a
single scale factor.
\end{proof}

The Kasner slicing of de Sitter (\ref{kasnerslice}) hints at how this anisotropy will be manifested; the $S^2$ factor scales differently than the ${\mathbb R}$ factor at small $t$.  Another simple example is formed by taking connected sum of two flat  ${\mathbb R}^3$ de Sitter initial data sets to form a designer de Sitter spacetime. This spacetime is simply connected. Furthermore, there is a local neighborhood of each point that evolves as a de Sitter spacetime. However, simple as its form is, its global dynamics will be more complicated. Intuitively, there is a length scale describing the size of a neighborhood beyond which isotropy is broken. Consequently an observer will make local measurements compatible with an isotropic spacetime on scales smaller than this length scale. On scales larger than this length scale, the observer will see anisotropy. More complicated examples can be constructed by taking connected sums of  flat and hyperbolic ${\mathbb R^3}/\Gamma$ initial data sets, spherical $S^3/\Gamma$ initial data sets and $S^2\times{\mathbb R}/\Gamma$ initial data sets.
In addition, there is no restriction on the number of factors in the connected sum.
Thus, there is basically an infinite number of free parameters in the set of designer de Sitter spacetimes:
those set by the global topology $\Sigma$ and others characterizing geometrical factors such as relative positions of the gluing regions in the geometrization of the connected sum.  
Thus a standard assumption in cosmology, that a single scale factor $a(t)$ describes the global evolution  of the universe is not true for the case of cosmological constant.  Therefore, designer de 
Sitter universes violate the Cosmological Principle in forms that assume global isotropy. 

\section{Conclusions}

There is no fundamental physical mechanism in cosmology that  rules out non-simply connected
topology such as $T^3$ whose dynamics are described by a single scale factor. Indeed such models have been studied previously \cite{Cornish:1997rp,
Cornish:1997hz,Weeks:1998qr}. However, there is clearly no reason to exclude other spatial topologies, in particular those of designer de Sitter universes. Hence, an interesting question is, what are potential observational signatures of designer de Sitter universes?

Local isotropy implies that physical processes that occur on relatively small scales will remain consistent with observations. For example, nucleosynthesis places strong constraints on the small scale anisotropy of the universe (see, for example, \cite{Rothman:1984vk}
 and references therein). However, these constraints cannot place strong bounds on designer de Sitter spacetimes as they are locally isotropic.

Observational signatures of finite size universes formed by identifications have been proposed. If the size of the universe is smaller than the surface of last scattering, these identifications lead to matched circles in the cosmic microwave background \cite{Cornish:1997rp,
Cornish:1997hz,Weeks:1998qr}. Simulations of this effect lead to limits on the identification scale \cite{Gausmann:2001aa,Riazuelo:2002ct} which are testable with the power spectrum results of WMAP. However, one expects that the observational signature of a universe with designer de Sitter topology will be quite different. These universes need not have any regularity as they are not formed by a group action on ${\mathbb R}^3$ or $S^3$.  Therefore they will not produce matched circles in the cosmic microwave background. Rather, their lack of global isotropy  would be expected to produce less regular deviations only at large length scales. As these universes exhibit an infinite number of parameters that characterize their topology and geometry, one can conjecture that by use of this freedom, one can match essentially any large scale feature in the cosmic microwave background by an appropriate choice of designer de Sitter universe. 

Finally, one could ask if this issue could be avoided by some physical mechanism. The dynamics of Einstein gravity do not allow topology change. The only  theory that can possibly  allow a prediction of the topology for the universe through dynamical effects is quantum gravity.
 However, it remains to be shown, once a consistent theory of quantum gravity is formulated in a way that this question can be answered, that in fact simple topology must occur. Thus for now, a better understanding of the observational consequences of designer de Sitter universes is most desirable.

\section*{Acknowledgements}
The work was supported by NSERC. In addition, the authors would like to
thank the Perimeter Institute for its hospitality during the writing of this paper.  Finally, DMW would like to thank the
organizers of Theory Canada III for their invitation to give this talk.

\end{document}